\def\year{2019}\relax
\documentclass[letterpaper]{article} %DO NOT CHANGE THIS
\usepackage{aaai19}  %Required
\usepackage{times}  %Required
\usepackage{helvet}  %Required
\usepackage{courier}  %Required
\usepackage{url}  %Required
\usepackage{graphicx}  %Required
\usepackage{enumitem}
\usepackage{hyperref}
\hypersetup{colorlinks, citecolor=blue, filecolor=blue, linkcolor=blue, urlcolor=blue}
\usepackage[amsthm]{ntheorem}
\usepackage{amssymb}
\usepackage{amsmath}
\usepackage{latexsym}
\newtheorem{theorem}{Theorem}
\newtheorem{lemma}{Lemma}
\newtheorem{example}{Example}
\newtheorem{definition}{Definition}
\usepackage{comment}
\usepackage{color}
\usepackage{xspace}
\usepackage{lipsum}
\usepackage{mathtools}
\usepackage{cuted}
\usepackage{comment}
\usepackage{caption}
\usepackage{subcaption}

\newcommand{\citet}[1]
{\citeauthor{#1}~\shortcite{#1}}
\newcommand{\citep}{\cite}

\newcommand{\xhdr}[1]{\vspace{2mm} \noindent{\bf #1}}

\newtheorem{question}{Question}

\newcommand{\bluee}[1]{{\color{blue} #1}}
\newcommand{\redd}[1]{{\color{red} #1}}

\newcommand{\COMM}[1]{}

\newcommand{\nr}[1]{E(#1)}  % the set of neighboring edges incident 

\newcommand{\cra}{\ensuremath{\operatorname{CR-ALG}}\xspace}
\newcommand{\lpa}{\ensuremath{\operatorname{MMP-ALG}}\xspace}
\newcommand{\MP}{\ensuremath{\operatorname{MMP}}\xspace}
\newcommand{\osbm}{\ensuremath{\operatorname{OSBM}}\xspace}
\newcommand{\obm}{\ensuremath{\operatorname{OBM}}\xspace}

\newcommand{\LP}{\ensuremath{\operatorname{LP	}}\xspace}

\newcommand{\OPT}{\ensuremath{\operatorname{OPT}}\xspace}
\newcommand{\ALG}{\ensuremath{\operatorname{ALG}}\xspace}

\newcommand{\NEGCR}{\ensuremath{\operatorname{NEG-CR}}\xspace}

\newcommand{\ie}{\emph{i.e.,}\xspace}
\newcommand{\eg}{\emph{e.g.,}\xspace}
\newcommand{\bI}{\mathbf{I}}

\newcommand{\cX}{\mathcal{X}}

\newcommand{\sE}{\mathbf{1}}   %attention please!

\newcommand{\cM}{\mathcal{M}}
\newcommand{\E}{\mathbb{E}}

\newcommand{\cP}{\mathcal{P}}
\newcommand{\cI}{\mathcal{I}}
\newcommand{\cR}{\mathcal{R}}

\newcommand{\bz}{ \mathbf{0}}

\newcommand{\supp}{Supp}

\newcommand{\x}{\mathbf{x}}
\newcommand{\y}{\mathbf{y}}

\newcommand{\ba}{\mathbf{a}}
\newcommand{\bb}{\mathbf{b}}

\newcommand{\X}{\mathbf{X}}
\newcommand{\Y}{\mathbf{Y}}
\newcommand{\Z}{\mathbf{Z}}

\newcommand{\ep}{\epsilon}
\newcommand{\alp}{\alpha}

\usepackage{algorithmic, algorithm}% http://ctan.org/pkg/algorithmicx

\begin{document}

\title{Balancing Relevance and Diversity in Online Bipartite Matching via Submodularity}

\author{John P. Dickerson, Karthik Abinav Sankararaman, Aravind Srinivasan, Pan Xu\\
	\{john, kabinav, srin, panxu\}@cs.umd.edu \\
	University of Maryland, College Park, MD, USA\\ \\
\emph{``Different roads sometimes lead to the same castle.'' -- George R.R. Martin}}

\maketitle

\vspace{2mm}
\begin{abstract}
\vspace{-2mm}
In bipartite matching problems, vertices on one side of a bipartite graph are paired with those on the other.  In its online variant, one side of the graph is available offline, while the vertices on the other side arrive online. When a vertex arrives, an irrevocable and immediate decision should be made by the algorithm; either match it to an available vertex or drop it.  Examples of such problems include matching workers to firms, advertisers to keywords, organs to patients, and so on.  Much of the literature focuses on maximizing the total relevance---modeled via total weight---of the matching. However, in many real-world problems, it is also important to consider contributions of diversity: hiring a diverse pool of candidates, displaying a relevant but diverse set of ads, and so on.  In this paper, we propose the Online Submodular Bipartite Matching (\osbm) problem, where the goal is to maximize a submodular function $f$ over the set of matched edges. This objective is general enough to capture the notion of both diversity (\emph{e.g.,} a weighted coverage function) and relevance (\emph{e.g.,} the traditional linear function)---as well as many other natural objective functions occurring in practice (\emph{e.g.,} limited total budget in advertising settings). We propose novel algorithms that have provable guarantees and are essentially optimal when restricted to various special cases. We also run experiments on real-world and synthetic datasets to validate our algorithms.
\end{abstract}

%%%%% Old Abstract %%%%%
% Online bipartite matching problems have been a source of significant interest due to the ubiquity of matching markets in practice, \emph{e.g.,} matching jobs to candidates, advertisers to keywords, organs to donors, to name a few. In many of these markets there are two competing objectives, namely, relevance and diversity. Much of the literature has studied problems where the goal is to maximize the total relevance of the matching (modeled as the total weight of the matching). However, in several such markets an equally important objective is also to maximize an appropriate notion of diversity (\emph{e.g.,} hiring a diverse pool of candidates for jobs). In this paper, we propose a model called the Online Submodular Bipartite Matching (\osbm) where, similar to the online bipartite matching problem, one set of vertices are available offline while the vertices from the other set come online and the algorithm has to make an irrevocable and immediate match. However, in this problem the goal is to maximize a submodular function $f$ over the set of matched edges. This objective is general enough to capture the notion of both diversity (\emph{e.g.,} weighted coverage function) and relevance (\emph{e.g.} linear function) as well as many other natural objective functions occurring in practice (\emph{e.g.,} limited total budget in advertising setting). We propose novel algorithms that have provable guarantees and are essentially optimal when restricted to various special cases. We also run numerical experiments to validate our algorithms.

\section{Introduction}\label{sec:intro}

%\bluee{the first paragraph should be removed and add one or sentence to the second one}. 
Online Bipartite Matching (\obm) problems are primarily motivated by Internet advertising. In the basic version of this problem, we are given a bipartite graph $G = (U, V, E)$, where $U$ and $V$ represent the offline vertices (advertisers) and online vertices (keywords or impressions) respectively. An edge $e=(u,v)$ represents a bid by advertiser $u$ for a keyword $v$. When a keyword $v$ arrives, a central agency must make an \emph{instant and irrevocable} decision to either reject $v$ or assign $v$ to one of its ``neighbors'' (\ie a vertex that is connected to $v$ by an edge in $G$) $u$ and obtain a profit $w_e$ for the match $e=(u,v)$. A matched advertiser $u$ is no longer available for future matches. The goal is to design an efficient online algorithm that maximizes the expected total weight (profit) of the matching. Following the seminal work of~\citet{kvv}, there has been a large body of research on related variants~%\citep[see, \eg,][]{mehta2012online}.  % <-- use citep to force [Author, Name] ....?
\cite{mehta2012online}.
During the last decade, \obm and its variants have seen wider applications in various matching markets: crowdsourcing Internet marketplaces~(\eg \cite{assadi2015online,ho2012online}), online spatial crowdsourcing platforms~(\eg \cite{TongVLDB17,TongICDE16,tong2016vldb}), ride-sharing platforms~(\eg \cite{DickersonAAAI18}). In each of these applications we have the following features (1) agents from at least one side appear online (\ie one-by-one) (2) an online agent on arrival has to either be matched immediately to an offline agent---or be rejected.

 % Online bipartite matching was first introduced by \cite{kvv} and attracted more attention later due to the Internet advertising business (see the survey \cite{mehta2012online}).
 
 %\bluee{add more application papers here into other areas, related to AI community}. 
% share two distinctive features which are uniquely associated with online matching models
%typically captured by a linear function.

Most prior research on online matching focuses on maximizing the total weight of the final matching~\cite{mehta2012online}, which captures the quality/relevance of all the matches. In many matching markets, we also care about the \emph{diversity} of the final matching along with relevance. \citet{Faez-17} considered a motivating example of matching academic papers to potential reviewers: \emph{just} maximizing the relevance (the quality of each match) could potentially assign a paper to multiple scholars in a single lab due to shared expertise, which is undesirable. Instead, we want to assign each paper to \emph{relevant} experts with diverse backgrounds to obtain comprehensive feedback. Maximizing diversity\footnote{Both individual and aggregate diversity~\cite{ad12}.} is of particular importance in various recommendation systems, ranging from recommendations of new books and movies on eBay~\cite{Chen-TKDE} to returning search-engine queries~\cite{agrawal2009diversifying}. A common strategy to address diversity is to first formulate a specific objective (typically maximization over a submodular function\footnote{See Section~\ref{sec:prelim} for a formal definition and canonical examples.}) capturing the balance of diversity and relevance and then design an efficient algorithm---typically a greedy one---to solve it (\eg \cite{Faez-17} and references within).

%In all these scenarios, users hope to receive as much broad and diverse recommendations as possible. The most common technique to address diversity is to first formulate a specific objective (typically maximization over a submodular) capturing the balance of diversity and relevance and then design an efficient algorithm (typically greedy) to solve it. \cite{Faez-17} proposed a diverse weighted bipartite $b$-matching problem with minimization over a supmodular. \cite{Sha-16} formulated the tradeoff as a combinatorial optimization with maximization over a well-designed 
%submodular function. \cite{Qin-13} used the entropy regularizer to tackle the diversity and showed that the objective satisfied non-decreasing monotonicity and submodularity. Among all approaches, maximization over a coverage  function is the most common strategy, see, \eg \cite{Ziegler-05, Ge-10,Puthi-16}. 

Inspired by the broad applications of \obm and submodular maximization, we propose a variant of the online matching model which we call Online Submodular Bipartite Matching (\osbm). In particular, we answer the main Question~\ref{que:QMain}, defined formally below.

	\xhdr{Main model.} Suppose we have a bipartite graph $G=(U,V, E)$ where $U$ and $V$ represent the offline and online agents respectively. We have a finite time horizon $T$ (known beforehand) and for each time (or round) $t \in [T] \doteq \{1, 2, \ldots, T\}$, at most one vertex $v$ is sampled---in which case we say $v$ \emph{arrives}---from a given \emph{known} probability distribution $\{p_{v}\}$. That is, $\sum_{v \in V} p_{v} \le 1$; thus, with probability $1-\sum_{v \in V} p_{v}$, none of the vertices from $V$ will arrive at $t$. The sampling process is independent across different times. Let $r_v\doteq T\cdot p_v$ denote the expected number of
arrivals of $v$ in the $T$ online rounds (we interchangeably refer to this as the arrival rate of $v$). We assume that the value $r_v$ lies in $[0,1]$. 
%(we can create multiple copies of $v$ if $r_v>1$ and make the final arrival rate fall in $[0,1]$ for each copy). 
Once a vertex $v$ arrives, we need to make an \emph{immediate} and  \emph{irrevocable} decision: either to reject $v$ or assign $v$ to one of its neighbors in $U$. Each $u$ has a unit capacity: it will be unavailable in the future upon being matched.\footnote{The general case where each $u$ has a given capacity $C_u$ can be reduced to this by creating $C_u$ copies of $u$.} We are given a \emph{non-negative monotone submodular} function $f$ over $E$ as an input. Our goal is to design an online matching algorithm such that $\E[f(\cM)]$ is maximized, where $\cM$ is the final (random) matching obtained. 

	There are two sources contributing to the randomness of $\cM$: the stochasticity from the online arrivals of $V$ and the internal randomness used by the algorithm. Note that $\cM$ can be a semi-matching, where each $u$ has degree at most $1$ while some $v$ may have degree more than $1$ (due to multiple online arrivals of $v$). Following prior work~\cite{mehta2012online}, we assume $|V| \gg |U|$ and $T\gg 1$. Throughout this paper, we use edge $e=(u,v)$ and assignment of $v$ to $u$ interchangeably. 

	\begin{question}
		\label{que:QMain}
			Is there a constant-factor competitive ratio\footnote{See Definition \ref{defn:compRatio}. Constant refers to value being reasonably bounded away from zero even for large graphs.} for  online algorithm for the Online Submodular Bipartite Matching problem?  
	\end{question}
	\vspace{-2mm}
	\xhdr{Related model.} One important direction in addressing diversity in online algorithms has been via online convex programming. In particular,~\citet{agrawal2014fast} considered the model of maximizing a concave function under convex constraints. At each time-step a random vector is drawn from an unknown distribution and the goal is to satisfy a convex constraint in expectation. Our work differs from theirs in multiple aspects. First, the offline problem of \citet{agrawal2014fast} is poly-time solvable, while our problem even in the offline version has unknown hardness (status unknown for both NP- and APX-hardness). Equivalence between discrete and continuous functions exists for submodular minimization via the Lov\'{a}sz extension. However, a similar continuous relaxation for submodular maximization is NP-hard to evaluate.\footnote{\eg Slide 26 in https://goo.gl/HAhqaZ}  Hence it is unclear how one would use their model to address our problem.	 Secondly, they assume small budgets while all matching type problems differ from allocation problems in that this assumption is not true (in fact, the main challenge is small budgets). The other difference is that our ``known i.i.d.'' gives algorithm design more power as compared to unknown distributions and therefore helps obtain improved ratios rigorously. For example, the online matching problem with linear objectives has been studied both in unknown distribution and known i.i.d. models separately since it presents a natural trade-off---knowing more information about the distribution and the competitive ratio. Based on applications, one would make assumption one-way or the other.

\emph{Special cases.} Our model generalizes some well-known problems in this literature. Note that if the submodular function is just a linear function of the weights this reduces to online weighted matching. Our model can also capture the Submodular Welfare Maximization (SWM) problem~\cite{kapralov2013online}; this problem and its variants have been widely studied in machine learning and economics (\eg see \cite{nips16Hossein} and references within). Given an instance of SWM we can add polynomially many extra vertices and reduce it to an instance of our problem. Due to space constraints we defer the proof of this reduction and all other theorems/lemmas in this paper to the supplementary materials.

\setlength{\belowcaptionskip}{-10pt}
\begin{figure*}[h!]
    \centering
     \includegraphics[width=\textwidth]{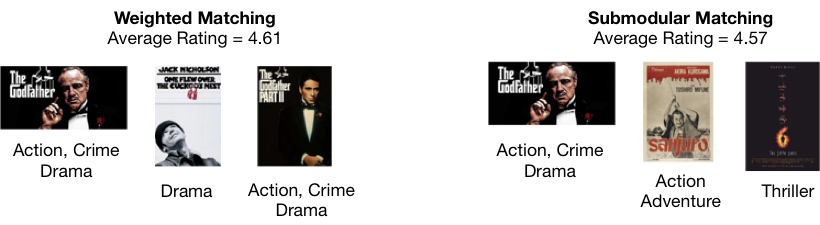}
      \caption{Recommended Movies for User 574. Left - weighted matching (top 3 highest predictions). Right - submodular matching with coverage function (balancing diversity of genres).}
       \label{fig:movies}
\end{figure*}

\xhdr{Applications.} We briefly describe some motivating examples for our problem. The first important application is in \emph{recommender systems}. Consider the problem of recommendation in platforms like Netflix, Amazon, etc. We have a set of users who come online and the system needs to choose a subset of movies, items to recommend to the user which has the most relevance. At the same time, the recommendations to any user needs to be diverse both from an engagement perspective (\eg, a user doesn't want to see only action movies recommended or just a single brand of items while shopping) and from a fairness perspective (\eg not showing only stereotypical recommendations based on race, gender, etc.). See Fig.~\ref{fig:movies} for an example on the MovieLens dataset. This naturally fits our model, where we can capture this trade-off between relevance and diversity via a weighted coverage function (which is monotone and submodular). Another application is in \emph{online advertising and auction design}. Retargeting in personalized advertisements is a major innovation in the past decade where potential advertisers collect background information to provide a better ad experience to their users. One of the major optimization problems for the advertiser is to create various ads to be shown based on the user profile without knowing a-priori the demand for various versions. Hence the advertiser instead specifies a single budget $B$ for a bundle of ads. The goal of the ad-matching agency is to run a matching algorithm with the objective that the revenue they will get for this bundle is $\min \{B, \sum_{e \in \cM} w_e$\}, which is a submodular function. Another application is in \emph{matching candidates to jobs in a dynamic job market}. A hiring agency announces openings for various positions and hires the best candidates as and when they come.  The most basic version of this problem is called the \emph{secretary problem}~\cite{vanderbei1980optimal}. The goal is usually to hire the best candidates for the open positions. Recently, with increasing awareness of various systemic biases, companies also look to hire diverse candidates (based on various metrics of diversity such as gender, race, and political leanings), which the classical secretary problem and its variants do not consider.\footnote{These have been explored practically in some recruiting systems~\cite{hong2013ihr}.} However, using our model with a submodular function such as a weighted coverage function over the various metrics as an objective, we can essentially capture hiring the best candidates while maximizing diversity.

\xhdr{Arrival assumption}.
The literature on Online Matching considers three broad classes of arrival assumptions: Adversarial Order (AO) (\eg \cite{assadi2015online}), Random Arrival Order (RAO) (\eg \cite{zhao2014crowdsource,subramanian2015online}), and Known Independent and Identical Distribution (KIID). In KIID, an online agent's arrival is modeled as a sample (identically with replacement) from a known distribution (\eg \cite{singer2013pricing,singla2013truthful}). In this paper, we consider the KIID assumption which captures the fact that the distribution over types can be learnt from historical data; thus can get improved ratios over other models (see \cite{DickersonAAAI18} for a discussion). 

\xhdr{Our contributions.}
	Our contributions can be summarized as follows. First, we propose the Online Submodular Bipartite Matching (\osbm) model, which abstractly captures the balance between \emph{relevance} and \emph{diversity} in the context of matching markets. Next, we provide two \emph{provably} good algorithms for this model \footnote{One of them works under a mild assumption of $|U| = o(\sqrt{T})$}. The first algorithm is based on Contention Resolution schemes used in the offline submodular maximization literature. This algorithm works for the case when the arrival rates are integral. Our second algorithm is based on using a feasible solution to an appropriate mathematical program, where this feasible solution approximates the offline optimal by a factor $1-1/e$, to guide the online actions. This algorithm works for the general case of \emph{arbitrary} arrival rates. The ratio achieved by this algorithm is \emph{tight} even when restricted to the special case of linear objectives, however the proof only works when the number of rounds $T \rightarrow \infty$. Nonetheless it can be seen as a natural generalization to submodular functions of the LP-based algorithm proposed by~\citet{haeupler2011online}. Finally, we run 
		experiments on both real-world as well as synthetic datasets on some common submodular functions to validate our algorithms and compare them to natural heuristics. 
	
%In this paper we propose a general framework on that: maximize a general monotone submodular (MS) in an online setting.
% Details the two features: matching: matching polytope, each offline vertices has capacity one WLOG; then KAD. 
%For the maximization of a monotone submodular (MS) subject to a simple cardinality constraint, \cite{nem-78} show
\xhdr{Related work.}
%the Maximization of a Monotone Submodular Function 
%in both Operations Research and Theoretical Computer Science (TCS) communities.
The offline version of our problem is the well-studied ``maximizing a monotone submodular function subject to a bipartite matching polytope constraint'' problem. More generally, the constraint set can be viewed as an intersection of two partition matroids. The general area of submodular maximization is well studied; here, we only survey algorithmic advances related to maximization of a monotone submodular function subject to various constraints. The classical work of \citet{nem-78} showed that the natural greedy algorithm achieves an $(1-1/e)$-approximation under a cardinality constraint, which is optimal in the value oracle model assuming $P\neq NP$~\cite{nembest-78}.  Under a general matroid constraint, \citet{cal-11} gave an algorithm achieving the optimal ratio of $1-1/e$ (in the value oracle model defined in Section~\ref{sec:prelim}) using the \emph{pipage rounding} technique. \citet{lee10} considered the constraint case of $k$ matroids with $k \ge 2$ and presented a \emph{local-search} based algorithm.\COMM{which achieves a ratio of $1/(k+\epsilon)$ for any $\epsilon>0$.} \citet{Kanthi-17} studied the case of intersection of $k$ matroids and a single knapsack constraint. \COMM{and gave a $\frac{1-e^{-(k+1)}}{k+1}$-approximation algorithm for any $k \ge 1$.}Recently a series of works has considered submodular maximization in the online setting. In particular, \citet{buchbinder-15} and \citet{chan2017} studied online submodular maximization in the adversarial arrival order with preemption: on arrival of an item, we should decide whether to accept it or not and \emph{possibly rejecting a previously accepted item}. In this paper, we do not allow preemption but consider a more flexible arrival assumption (\emph{i.e.,} KIID). This makes the problem tractable and admits algorithms with non-trivial competitive ratios. Apart from the offline and online models, submodular maximization has received much attention in other models due to its applications in summarization~\cite{tschiatschek2014learning}, data subset selection and active learning~\cite{wei2015submodularity}, and diverse summarization~\cite{mirzasoleiman2016fast}, to name a few. It has been studied in the streaming \cite{badanidiyuru2014streaming,mirzasoleiman2017streaming}, distributed \cite{mirzasoleiman2016fast,mirzasoleiman2016distributed} and stochastic~\cite{karimi2017stochastic,stan2017probabilistic} settings. Online Bipartite matching has been studied with a long line of work, overviewed comprehensively by~\citet{mehta2012online}. In the KIID arrival model, \citet{feldman2009online} introduced the idea of two suggested matchings and used that to guide the online phase, which was the first to beat $1-1/e$ for the \emph{unweighted} online bipartite matching. A similar idea was used by~\citet{haeupler2011online} for \emph{edge-weighted} case. \citet{manshadi2012online},~\citet{jaillet2013online}, and~\citet{brubach2016new} designed \emph{LP-based} online algorithms for the unweighted, vertex-weighted and edge-weighted versions of online matching problems to achieve the best known ratios. 

 \section{Preliminaries}
 \label{sec:prelim}

 We first describe the notation used throughout this paper. For two vectors $\ba$ and $\bb$,  $\ba \le \bb$ denotes the coordinate-wise ``$\leq$'' operation.  For a binary vector $\ba =(a_e) \in \{0,1\}^m$, let $\supp(\ba) =\{e: a_e=1\}$ be the support of $\ba$; we write $f(\ba)$ as a short-hand for $f(\supp(\ba))$ for any set function over $E$. In this paper, we use vectors to denote sets (\emph{i.e.,} using the indicator binary vector for a set). $e$ is used both as an edge index as well as Euler's constant; usage will be apparent from the context. We now give a formal definition of the submodular function and describe some canonical examples.

	\begin{definition}[Submodular function]
	A function $f: 2^{[n]} \rightarrow \mathbb{R}^{+}$ on a ground-set of elements $[n] := \{1, 2, \ldots, n\}$ is called submodular if for every $A, B \subseteq [n]$, we have that $f(A \cup B) + f(A\cap B) \leq f(A) + f(B)$ and $f(\phi) = 0$. Additionally, $f$ is said to be monotone if for every $A \subseteq B \subseteq [n]$, we have that $f(A) \leq f(B)$. 
\end{definition}

For our algorithms, we assume a \emph{value-oracle} access to a submodular function. This means that, there is an oracle which on querying a subset $T \subseteq [n]$, returns the value $f(T)$. The algorithm does not have access to $f$ \emph{explicitly}.

	\xhdr{Examples.} Some common examples of submodular functions include the coverage function, piece-wise linear functions, budget-additive functions among others. In our experiments section, we use the following two examples.
		\begin{enumerate}
			\item \emph{Coverage function.} Given a universe $\mathcal{U}$ and $g$ subsets $A_1, A_2, \ldots, A_g \subseteq \mathcal{U}$, the function $f(S) = |\cup_{i \in S} A_i|$ is called the coverage function for any $S \subseteq [g]$. This can naturally be extended to the weighted case. Given a non-negative weight function $w: \mathcal{U} \rightarrow \mathbb{R}^{+}$, then the weighted coverage function is defined as $f(S) = w(\cup_{i \in S} A_i)$.

			\item \emph{Budget-additive function.} For a given total budget $B$ and a set of weights $w_i \geq 0$ on the elements $[g]$ of universe $\mathcal{U}$, for any subset $S \subseteq \mathcal{U}$ the budget-additive function is defined as $f(S)=\min\{ \sum_{i \in S} w_i, B\}$.

		\end{enumerate}

	\begin{definition}[Multilinear extension]
		The multilinear extension of a submodular function $f$ is the continuous function $F: [0, 1]^n \rightarrow \mathbb{R}^+$ defined as 
		 $F(x) := \sum_{T \subseteq [n]} (\prod_{k \in T} x_k \prod_{k \not \in T}(1-x_k)) f(T)$.
	\end{definition}

Note that $F(\x) = f(\x)$ for every $\x \in \{0, 1\}^n$. The multilinear extension is a useful tool in maximization of submodular objectives. In particular, the above has the following probabilistic interpretation. Let $\mathcal{R}_{\x} \subseteq [n]$ be a random subset of items where each item $i \in [n]$ is added into $\mathcal{R}_{\x}$ independently with probability $x_i$. We then have $F(\x) = \mathbb{E}[f(\mathcal{R}_{\x})]$.
	
	\xhdr{Offline optimal.} Throughout this paper we use the terms offline optimal (or interchangeably offline optimal value) which refers to the following. Given a specific sequence $S$ of (random) arrivals, the ``offline problem'' is to find the best hindsight matching that maximizes the objective on this sequence, denoted by $\OPT(S)$. Note that $\OPT(S)$ is a random variable since $S$ is random. The offline optimal value, denoted by $\mathbb{E}[\OPT]$, is the expectation of $\OPT(S)$, where the expectation is taken over all possible random sequences $S$.
	
		\begin{definition}[Competitive ratio]
			\label{defn:compRatio}
		Let $\mathbb{E}[\ALG(\mathcal{I}, \mathcal{D})]$ denote the expected value obtained by an algorithm $\ALG$ on an instance $\mathcal{I}$ and arrival distribution $\mathcal{D}$. Let $\mathbb{E}[\OPT(\mathcal{I})]$ denote the expected \emph{offline optimal}. Then the competitive ratio is defined as $\min_{\mathcal{I}, \mathcal{D}} \mathbb{E}[\ALG(\mathcal{I}, \mathcal{D})]/\mathbb{E}[\OPT(\mathcal{I})]$.
	\end{definition}

 \section{Challenges and Main Techniques}

Our algorithm, like prior work on Online Matching, follows a two-phase approach divided into an Offline phase and an Online phase.
% Our model \osbm can be viewed as a significant generalization of edge-weighted \obm since linear objective is a very special submodular function. For edge-weighted \obm, a common strategy to claim that $\ALG$ achieves an online ratio of $\alpha$ is by edge-by-edge analysis: first show each edge is added into the final matching with marginal distribution at least $\alpha x^*_e$ where $x^*_e$ refers to the probability that $e$ is added in an offline optimal \OPT; then by \emph{linearity of expectation}, we can conclude that $\E[\ALG] \ge \alp \OPT$. However, this fails to work in our case when objective is maximization over a general monotone submodular. 
 
%We overcome the challenge in two ways.  
 \xhdr{Offline phase.}
The first key challenge is to obtain a good handle on the optimal offline solution. For the edge-weighted \obm, the offline version reduces to a maximum weighted matching problem on a bipartite graph, which can be solved efficiently. For \osbm, the offline version which is to maximize a general non-negative monotone submodular function within a bipartite polytope, is non-trivial. Neither polynomial time- nor APX-hardness of this problem is well-understood. We tackle this challenge by first proposing a (offline) Multilinear Maximization Program (\MP) where we maximize the multilinear extension $F$ of the given submodular function $f$ subject to bipartite matching constraints, and then use the continuous greedy algorithm~\cite{cal-11} to solve it. This gives us a marginal distribution $\x^*=\{x_e^*\}$ for each edge being added in the offline optimal satisfying $F(\x^*) \ge (1-1/e)\E[\OPT]$.

%into the final matching by
%The second challenge is how to claim the final online ratio based on the approximate marginal distribution $\x^*$. 

\xhdr{Online phase}. The next challenge is to use the approximate offline marginal distribution $\x^*$ to guide the online phase. We propose two online algorithms which take $\x^*$ as input and make the online decisions by using a modified version of this. The first is inspired by Theorem 4.3 due to~\citet{bansal2012solving} and its extension. We call this the \emph{CR-based} algorithm which works for the special case of \emph{integral arrival rates}. If $f$ is a linear function, then showing each edge $e$ is added by an online algorithm \ALG with probability at least $\alp x^*_e$ implies that $\ALG$ achieves a final ratio of $\alpha (1-1/e)$ (the second factor $1-1/e$ accounts for the loss in the offline phase). This is because $\E[\ALG] \ge \alp F(\x^*)$ by \emph{linearity of expectation}. However, this approach fails when $F$ is a multilinear extension of the submodular function $f$. We overcome this by using Theorem 4.3 of~\citet{bansal2012solving} (see supplementary materials for the theorem statement), which gives a sufficient condition to ensure that $\E[\ALG] \ge \alp F(\x^*)$ for any multilinear extension $F$, provided that each edge $e$ is added in \ALG with a marginal distribution at least $\alp x_e^*$. This framework is generalized to Contention Resolution (CR) schemes~\cite{vondrak2011submodular}, which is used as a tool for the following general problem. Consider a fractional $\x  \in \cP_{\cI}$ where $\cP_{\cI}$ is a convex relaxation of an integral polytope $\cI$  and let $\X=(X_e)$ (not necessarily within $\cI$) be a random indicator vector where every $X_e$ is a Bernoulli random variable with mean $x_e$. Our goal is to round $\X$ to another integral vector $\Y$ such that (1) $\Y \in \cI$ and (2) $\E[f(\Y)] \ge \alp \E[f(\X)]=\alp F(\x)$ with as large $\alp$ as possible.

 Our second proposed algorithm, which works for arbitrary arrival rates, is an \emph{\MP-based} algorithm (\lpa) described as follows. When a vertex $v$ arrives, sample a neighboring edge $e=(u,v)$ with probability $x_e^*$ and include it iff $u$ is still available. When $f$ is linear,~\citet{haeupler2011online} gave a simple and tight analysis\footnote{Here, ``tight'' refers to the analysis and not the problem formulation itself.} showing that $\lpa$ loses a factor of $1-1/e$ in the online phase. The tight example is as follows.
 
 \begin{example}
Consider an unweighted bipartite graph $G=(U,V,E)$ consisting of a perfect matching with $|U|=|V|=T$ with $x_e^*=1$ for each $e$ and $f(\x^*)=\sum_e x^*_e$. Notice that each $e=(u,v)$ is added by \lpa iff $v$ comes at least once during the $T$ rounds, which occurs with probability equal to $1-1/e$. Thus in this example, we have that $\E[\lpa] =(1-1/e) f(\x^*)$. 
%Thus we see that each edge will be added with probability $1-1/e$ if 
 \end{example} 

 In this paper, we give a tight analysis showing that \lpa losses a factor at most $1-1/e$ in the online phase even for an arbitrary non-negative monotone submodular function $f$.  The downside is that the bounds hold only in the limit when $T \rightarrow \infty$. We believe a careful modification of our proof will lead to a finite time analysis, but do not do so in this paper. In particular when $T \rightarrow \infty$, we prove that $\E[\lpa] \ge (1-1/e)F(\x^*)$ and thus this yields a final ratio of $(1-1/e)^2$ (after incorporating another factor of $1-1/e$ in the offline phase). The main proof idea is through a virtual algorithm \ALG which has the same performance as \lpa and by applying pipage rounding \cite{ageev2004pipage1} to \ALG we show that $\E[\ALG] \ge F((1-1/e)\x^*) \ge (1-1/e) F(\x^*)$.

\section{Offline Phase} 
\label{sec:offline-phase}

For an edge $e$, let $x_{e}$ be the probability that $e$ is chosen in any fixed offline optimal algorithm. For each $u$ (likewise for $v$), let $\nr{u}$ ($\nr{v}$) be the set of neighboring edges incident to $u$ ($v$) in $G$. Let $F: [0,1]^m \rightarrow \mathbb{R}_{+}$ be the multilinear extension of $f$. Consider the following mathematical program.

\vspace{-5mm}
\begin{alignat}{2}
\label{LP:offline}
\text{maximize}    & ~~\textstyle F(\x)  & \\
\text{subject to}  & \textstyle \sum_{e \in \nr{v} } x_{e} \le  r_v    & \textstyle \qquad \forall v \in V\label{cons:v}\\
                & \textstyle \sum_{e \in \nr{u} } x_{e} \le  1    &\textstyle  \qquad \forall u \in U\label{cons:u}\\
                   & \textstyle 0 \le x_{e} \le 1   & \textstyle \forall e \in E \label{cons:e}
\end{alignat}

The constraints can be interpreted informally as follows. Constraint \eqref{cons:v} states that the expected number of matches for any $v$ is no more than the expected number of arrivals of $v$ (\emph{i.e.,} $r_v$). Constraint \eqref{cons:u} states that the expected number of matches for every $u$ is no more than $1$, since $u$ has an unit capacity. Constraint \eqref{cons:e} is valid since every $x_e$ is a probability value.

\vspace{0.1in}
\begin{lemma}\label{lem:lp}
There is an efficient algorithm (running in polynomial time) which returns a feasible solution $\x^*$ to the program \eqref{LP:offline} such that $F(\x^*) \ge (1-1/e) \E[\OPT]$, where $\E[\OPT]$ is the offline optimal value.
 \end{lemma}

\section{Online Algorithms}\label{sec:online}

 In this section, we present several online algorithms that take $\x^*$ as an input, which is a feasible solution to the program \eqref{LP:offline}  with $F(\x^*) \ge (1-1/e)\E[\OPT]$, where $\E[\OPT]$ refers to the offline optimal. 
 
 \xhdr{A CR-based online algorithm.}\label{sec:cra}
 In this section, we present a CR-based online algorithm ($\cra$) for \osbm with integral arrival rates. In this case, we assume that $p_v=1/T$ and $r_v=1$ for all $v$.  
 
The main idea is as follows. We start with $\cR_{\x^*}$, which is obtained by independently sampling each edge $e$ with probability $x^*_e$. Let $\X \in \{0,1\}^m$ be the integral vector corresponding to $\cR_{\x^*}$ such that $X_e=1$ iff $e \in \cR_{\x^*}$. Let $E_{\X}(v) =\{e: e \in \nr{v}, X_e=1\}$ and $E_{\X}(u) =\{e: e \in \nr{u}, X_e=1\}$ be the set of sampled edges incident to $u$ and $v$ respectively. Now we obtain another random vector $\Y \in \{0,1\}^m$ from $\X$ by uniformly sampling an edge from $E_{\X}(u)$ for each $u$ (the sampling process is independent across different $u$). We then use both $\X$ and $\Y$ to guide the online phase. Algorithm~\ref{alg:nadap} describes this algorithm formally.
 
     \setlength{\textfloatsep}{5pt}
  \begin{algorithm}[h!]
\caption{A CR-based algorithm (\cra)} 
\label{alg:cr}
\begin{algorithmic}

\STATE{\textbf{Offline Phase:}}
 \STATE Solve Program  \eqref{LP:offline} using continuous greedy and let $\x^*=(x_e^*)$ be an approximate solution with $F(\x^*) \ge (1-1/e)\E[\OPT]$. 
 \STATE Independently sample each edge with probability $x^*_e$. Let $\X =(X_e)\in \{0,1\}^m$ be the resultant indicator vector such that $X_e=1$ iff $e$ is sampled.
 
 \STATE For each $w \in U\cup V$, let $E_{\X}(w) =\{e: e \in \nr{w}, X_e=1\}$ be the set of sampled edges incident to $w$. Sample one edge uniformly at random from $E_{\X}(u)$  for each $u$ if $E_{\X}(u) \neq \emptyset$. Let $\Y \le \X$ be the indicator vector of the final edges sampled.
 \STATE{\textbf{Online Phase:}}\\
When $v$ arrives at time $t$,
sample an edge $e$ uniformly from $E_{\X}(v)$. Match it if $Y_e=1$ and $e=(u,v)$ is safe at $t$ (\ie $u$ is available); skip it otherwise. 
\end{algorithmic}
\end{algorithm}

\begin{theorem} \label{thm:cra}
There exists an online algorithm $\cra$, which achieves an online competitive ratio of at least $\frac{1}{2}(1-e^{-1/2})(1-1/e)$ for \osbm with integral arrival rates.
\end{theorem}

\xhdr{An \MP-based online algorithm.}\label{sec:lpa}
  In this section, we present a \MP-based online algorithm (\lpa) for \osbm. Compared to \cra, it also extends to the regime of arbitrary arrival rates $r_v \in [0, 1]$ for each $v \in V$. Algorithm~\ref{alg:nadap} describes it formally.
   \setlength{\textfloatsep}{5pt}
   \begin{algorithm}[h!]
\caption{An \MP-based online algorithm (\lpa)} 
\label{alg:nadap}
\begin{algorithmic}

\STATE{\textbf{Offline Phase:}}
 \STATE Solve Program  \eqref{LP:offline} using continuous greedy and let $\x^*=(x_e^*)$ be an approximate solution with $F(\x^*) \ge (1-1/e)\E[\OPT]$. 
 
\STATE{\textbf{Online Phase:}}
 \STATE When $v$ arrives at time $t$, sample an edge $e$ from $E(v)$ with probability $\frac{x^*_e}{r_v}$ (at most one such edge gets sampled). Match it if  $e=(u,v)$ is safe at $t$ (\ie $u$ is available) and skip it otherwise. 
\end{algorithmic}
\end{algorithm}

\begin{theorem} \label{thm:lpa}
There exists a \MP-based online algorithm \lpa which achieves a competitive ratio of at least $(1-1/e)^2$ for \osbm when $|U|=o(\sqrt{T})$ and  $T \rightarrow \infty$.
\end{theorem}

\section{Experiments}
\label{sec:exp}

	In this section, we describe the experimental results for the movie recommendation application. Additional experiments using synthetic data on other submodular functions are relegated to the supplemental material. We use the MovieLens dataset \cite{harper2016movielens} for our purposes\footnote{Full code can be found at \url{https://bitbucket.org/karthikabinav/submodularmatching/src/master/}}. 
	
	\xhdr{Application.} We have a list of movies each associated with some genres and we have a set of users who come into the system at various times (\emph{e.g.,} a user logging into the website for a session). We have information about the ratings of every user for some (different) movies. Our goal is to recommend a small set of movies which the user hasn't rated thus far such that the list is relevant as well as diverse. We quantify the \emph{diversity} by using a weighted coverage function over the set of genres for each user. Hence the goal is to maximize the sum of these weighted coverage functions for all users.\footnote{The sum of submodular functions is also submodular.} This naturally fits within the framework proposed in this paper, where the set of movies form the side $U$ and the set of users form the side $V$.\footnote{See supplementary materials for details on a Linear Programming formulation of the offline problem.}
	
	\xhdr{Dataset and pre-processing.} In this dataset, we have 3952 movies, 6040 users and a total of 100209 ratings of the movies by the users. We choose 200 users who have given the most ratings and sub-sample 100 movies at random; this reduced dataset is used for experiments. Similar to~\citet{Faez-17}, we use a standard collaborative filtering approach \cite{collabFiltering} to complete the matrix of ratings. To create the graph we do the following. For a given pair of user $u$ and movie $m$, if $u$ hasn't rated $m$ we add an edge between $u$ and $m$. For a given user $u$, we compute the average predicted rating for every genre and use this average as the ``weight'' in the weighted coverage function. This gives a bias towards genres which the user has highly rated over ones they haven't. For every user we choose a random arrival probability (ensuring that the sum of arrival probabilities equals $1$).
	
	\xhdr{Algorithms.} We test both our CR-based (Algorithm~\ref{alg:cr}) and the \MP-based (Algorithm~\ref{alg:nadap}) experimentally. Additionally we consider the following two heuristics for our study. (1) \emph{Greedy:} At time step $t$, let $S_t$ be the set of edges chosen so far. When a vertex $v$ arrives, choose an available neighbor $u$ that maximizes $f(S_t \cup \{(u, v)\})-f(S_t)$. If no neighbor is available, we drop $v$. (2) \emph{Negative CR-based algorithm (\NEGCR):} We tweak \cra by replacing the initial independent sampling with the following procedure. At every $u$, we use the dependent rounding routine due to \citet{gandhi2006dependent} to obtain a semi-matching $\mathcal{M}_1$. In the online phase when a $v$ arrives, we sample one of its available neighbors in $\mathcal{M}_1$ uniformly and match it. If not, we drop $v$.
	
	\xhdr{Results and discussion.} We run two kinds of experiments for our purposes. First, we compare our algorithms against the baselines by varying two parameters $B$ and $\eta$. $B$ represents the number of times we can match a movie to an user and $\eta$ represents the number of movies matched to any user on arrival (in the theory $B=1, \eta=1$, but we experiment with different values). Second, we want to \emph{measure} diversity of recommendations of the various algorithms. To this end, we compare the various algorithms on the number of users who have various levels of \emph{coverage} (\ie how many users are shown recommendations greater than x \% of the total weight). The plots in Figure~\ref{fig:Mexp1} and the leftmost plot of Figure~\ref{fig:Mexp2} show the results for the first kind of experiments, while the right two plots in Figure~\ref{fig:Mexp2} show the results for the second kind.

		\begin{figure*}[!h]
			\centering
			\includegraphics[width=\linewidth]{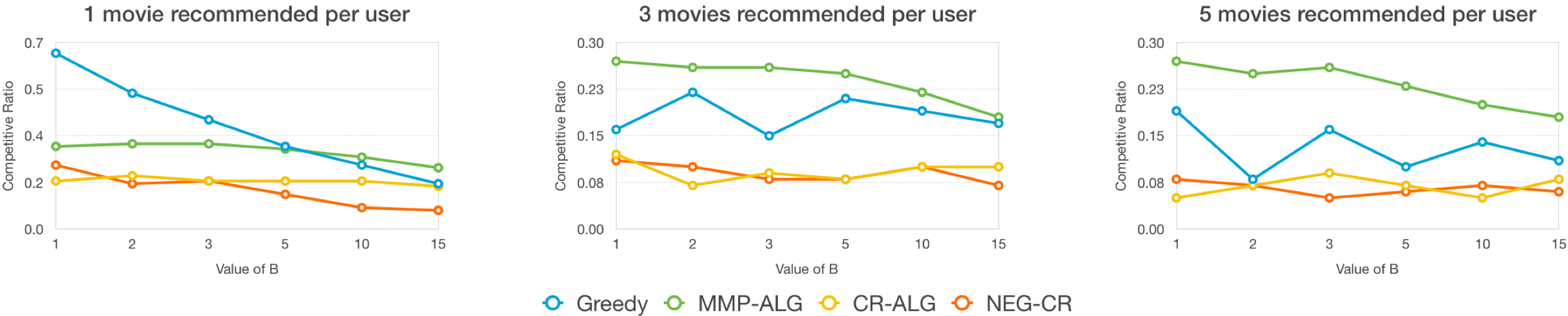}
			\caption{Results when the genre weights are the average of predicted ratings for users. The $x$-axis varies $B$ and the $y$-axis represents the ratio. (Left): $\eta=1$, (Center): $\eta=3$, (Right): $\eta=5$.}
			\label{fig:Mexp1}
		\end{figure*}
		
		\begin{figure*}[!h]
			\centering
				\includegraphics[width=\linewidth]{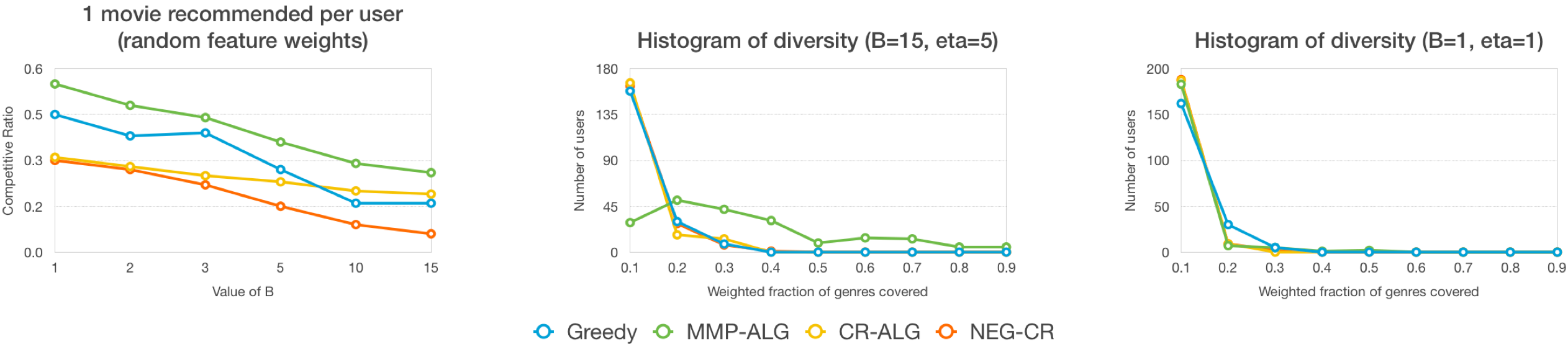}
				\caption{(Left): Same as the left plot in Figure~\ref{fig:Mexp1} with genre weight chosen $U[0, 1]$. (Center and Right): $x$-axis percentage of coverage of genres and $y$-axis number of users who fall in that range.}
				\label{fig:Mexp2}
		\end{figure*}

		In almost all cases, these plots show that \lpa is the clear winner and has the best performance. At times the Greedy algorithm does well, but as $B$ increases the performance drops quickly. Additionally, Greedy makes many calls to the submodular oracle at each online step, as compared to the other algorithms, which can be limiting in if the oracle evaluation is time-consuming. The surprising aspect on these experiments is that the other proposed algorithm $\cra$ does not perform even as well as Greedy. The explanation is that we assign non-integral arrival rates to each user. We show in the supplementary materials that when they are assigned integral rates, \cra's performance is comparable to \lpa and much better than Greedy. As we further show in the supplementary materials, for the budget-additive submodular function, however, even for fractional rates, \cra performs as well as \lpa (and much higher than the theoretical bounds). The diversity histograms show that \lpa is performs well and a good fraction of the users have a coverage greater than 50 \% and all the way up to 90\% (in the pragmatic case of $b=15, \eta=5$). However, the other algorithms have a coverage of at most 20 - 30\% for all users (even for the Greedy algorithm when $B=1, \eta=1$ where the competitive ratio is higher than \lpa).

	%%%%%% Conclusion %%%%%%%
	\vspace{-3mm}
	\section{Conclusion}
		\label{sec:conclusion}
		
		In this paper, we proposed a new model, Online Submodular bipartite matching (\osbm), which effectively captures notions such as relevance and diversity in matching markets. Many applications such as advertising, hiring diverse candidates, recommending movies or songs naturally fit within this framework. We propose two algorithms, one based on contention-resolution schemes and the other based on using the solution of the mathematical program directly; we give theoretical guarantees on their performance. The algorithm based on using the mathematical program directly is essentially tight even for the special case of linear objectives. Finally, via experiments we show that our algorithms do well in practice. We also proposed heuristics, some of which perform well on specialized submodular functions, and showed that our general algorithm is competitive with such algorithms as well.

%%%%%%%%%%%%%%%% Rough Notes from before %%%%%%%%%%%%%%%%%%%%%%%%

\COMM{
Recall that for the Ad words application. We have one advertiser. We have a bipartite graph $G=(U,V,E)$, where each $u$ represent a type of ads, which includes the specific information about an item (such as color) and the display venue, for example: a blue Toyota Camry displayed on Facebook. Each $v$ denotes a type of user. Now for simplicity still assume each $u$ has capacity $1$ and each edge $e=(u,v)$ has a bid $b_e$ and we have a whole budget $B$. In this context, the objective function should be 
$$f(S) =\min \big(B, \sum_{e \in S} b_e \big)$$ 
where $S$ forms semi-matching over $E$. Notice that the above model surely generalizes the classical online matching when $B=\infty$. Note that solving a general MS function over a bipartite matching polytope, we do not know the best result. Since the latter is an intersection of two matroid, there is some ratio $\frac{1}{2+\ep}$ for any given $\ep$ from local search. So far we are unsure if we can beat it for our special case: bipartite matching polytope. A typical approach is first solve the multilinear extension and we get a fractional $\x$ and loss a factor $1-1/e$ and then try to round $\x$ to an integral $\X$ (potentially losing some more?).  

\bluee{But for our online case, we just need a feasible fractional solution $\x$ in the bipartite matching polytope, such that $F(\x^*) \ge c \max_{\x \in \cP} F(\x) \ge c \max_{\x \in \cP \cap \{0,1\}^m } F(\x)=c  \cdot \OPT$}. 

\redd{In our special case when $f$ is defined as above, we can get a fractional feasible solution $\x^*$ such that $f(\x^*) \ge \OPT$. Notice that in our case, $f$ itself can be viewed as an extension over $[0,1]^m$. Here we just treat $\x^*$ can serve as a good online guide: the online theoretical analysis does not apply here}. In fact, in our case, if $f(\x^*)=B$, then we can continue by linear expectation and do an edge by edge analysis; If $f(\x^*) \le B$, the same goes. What I want to stress here is for either case, if we can prove that $\Pr[e \text{ is finally chosen}] \ge c*x_e$, we can not claim that we get an online ratio of $c$ as theoretical analysis. 

$$\left\{ \max y: y \le B, y=\sum_{e \in E} x_e b_e, x_u \le 1, \forall u, x_v \le 1, \forall v  \right\}$$
  
Consider a given feasible solution $\x$ to the above program, we see that $y(\x) \ge F(\x)$, where $F$ is the ME (multilinear extension) of $f$. In fact for any feasible solution $\x \in \cP$, suppose we define $y(\x)=\min \Big(B, \sum_{e \in E} x_e b_e \Big)$. Then we have that $y(\x) \ge F(\x)$. 

\begin{proof}
Consider the first case when $\sum_{e \in E} x_e b_e < B$. In this case we see $y(\x)=\sum_{e \in E} x_e b_e$. In contrast we see that $F(\x) \le \mu=y(\x)$. 
Let $X=\sum_{e \in E} b_e\cdot \mathbf{I}_e$, where $\{\mathbf{I}_e\}$ are indep and each is a Bernoulli with mean $x_e$. Then we have 
$$F(\x)=\Pr[X \le B] \E[ X | X \le B]+\Pr[X >B]\cdot B$$

Consider another case when $\sum_{e \in E} x_e b_e \ge B$. In this case we see 
$y(\x)=B \ge F(\x)$.
\end{proof}

Thus we see that in our context, MS itself can be viewed as an extension over $[0,1]^m$ but slightly larger than that of $F(\x)$. Another interesting fact is that. Starting from $\x^*$ to the \LP above, suppose we prove that if we can prove that $\Pr[e \text{ is finally chosen}] \ge c*x_e$, we can not claim that we get an online ratio of $c$ as before. Our final goal is to prove $\E[f(\Z)] \ge c y(\x)$ where $\Z$ is the final integral vector after incorporating online arrivals and algorithm. For the \LP-based, we have that $\Pr[Z_e=1] \ge (1-1/e)x_e$. So can we claim for the $c$? We do not know but negative correlation surely helpful than positive correlation. I think the WS arrival at $y(\x)=B$ and so the overflow can be maximized. Suppose by attenuation we can ensure $\Pr[Z_e]=(1-1/e)x_e$, so what? Apply the result in online submodular when $|U|=o(\sqrt{T})$ may be helpful and we can continue to apply Chernoff bound. So after considering more real assumptions and design more proper algorithms.  
such as each $b_e \le 1$ or is there any assumption for $b(u)=\sum_{e \sim u}b_e$ and also most importantly after consider the $B'$ matching: each $u$ can be matched at most $B'$ times and surely we should discuss the relation with $B'$ with $\sum_{u } b(u)$? \bluee{note that our model featuring the global budget $B$ can NOT be viewed as special case of Adwords in case each $u$ has a budget}. 

Aother application is as follows: each $u$ is a cluster which consists of $3$ copies, each $u$ denotes a job while each online $v$ is a candidate, which is associated with a vector $\mathbf{w}_v \in {0,1}^K$, where $K$ is the set of attributes. for example (Black, female, Asian, Sexual orientation). Each $u$ has three copies means we have $3$ opening for that. Suppose we want to maximize the diversity of all candidates enrolled for each job $u$. Let $w_k$ be the weight of attribute $k \in K$. For each cluster $u' \in C_u$ should have the same set of neighbors in $V$. Consider a given fractional $\x$. Let $x_{u',k}=\sum_{v \sim u'} x_{u',v} w_{v,k}$, the contribution of coverage from $u'$ to the first attribute. Thus we have 
$$\max  \sum_{u \in U}\sum_{k \in [K]}\min \big( \sum_{u' \in C_u} x_{u',k}, 1\big) w_k$$    

where $\tau(u)=\sum_{k \in [K]}\min \big( \sum_{u' \in C_u} x_{u',k}, 1\big) w_k$ quantifies the diversity associated with job $u$. We hope for each $u$, all candidates can be diversified as much as possible. Similarly as before we can solve a $\x$ such that $f(\x) \ge \OPT$ while $f$ itself can be viewed as an extension over $[0,1]^m$. 
}

\section*{Acknowledgements}
Aravind Srinivasan's research was supported in part by NSF Awards CNS-1010789, CCF-1422569 and CCF-1749864, and by research awards from Adobe, Inc. The research of Karthik Sankararaman and Pan Xu was supported in part by NSF Awards CNS 1010789 and CCF 1422569.

 The authors would like to thank the anonymous reviewers for their helpful feedback.

\bibliographystyle{aaai}
\bibliography{refs_short}

\newpage
\onecolumn
\section{Supplementary Materials}
\subsection{Details on prior work on submodularity}
As discussed in the preliminaries of the main section, we incorporate ideas from prior work on submodularity. In particular we use the following theorem from~\cite{bansal2012solving} and its extension called contention-resolution schemes~(\cite{vondrak2011submodular}).

\begin{theorem}(\cite{bansal2012solving})\label{thm:toc}
Let $f$ be a non-negative monotone submodular function over $E$ with $|E|=m$. Given a fractional vector $\x \in [0,1]^m$, let $\X=(X_e)_{e \in E}$ be a random binary vector such that each $X_e$ is a Bernoulli random variable with mean $x_e$. Suppose $\Z \le \X$ is another random binary vector satisfying the following properties (1) \& (2) for some $\alp$. Then we have that $\E[f(\Z)] \ge \alpha \E[f(\X)] $.

\begin{enumerate}
	\item \textbf{Marginal Property}: $\Pr[Z_e=1 | X_e=1] \ge \alp$.\\
	\item \textbf{Monotonicity Property}: $\Pr[ Z_e=1 | \X=\ba] \ge \Pr[ Z_e=1 |\X=\bb]$, $ \forall \mathbf{a} \le \mathbf{b} \in \{0,1\}^m, \forall e\in \supp(\ba)=\{e': a_{e'}=1\}$.
%$ \forall \mathbf{a} \le \mathbf{b} \in \{0,1\}^m, \forall e\in \supp(\ba)=\{e': a_{ecite'}=1\}:~~\Pr[ Z_e=1 | \X=\ba] \ge \Pr[ Z_e=1 |\X=\bb]$$
%Notice that by definition of multilinear extension, we have $\E[f(\X)]=F(\x)$.
\end{enumerate}

\end{theorem}
 By definition of the multilinear extension we have that $\E[f(\X)]=F(\x)$. Thus the theorem can be viewed as conditions when ``linearity of expectation'' can be applied to $F$.

\subsection{Submodular Welfare Maximization (SWM) as a special case of our problem}
	We will now show how we can reduce the SWM problem to our model. Recall that the SWM problem is defined as follows. We have $n$ vertices in $U$ which are available offline. With each vertex $u \in U$ we have a monotone submodular function $g_u: 2^V \rightarrow \mathbb{R}$ associated with it. When a vertex $v$ arrives we need to match $v$ to one of its neighbors in $U$. At the end of the online phase, let $S_1, S_2, \ldots, S_n$ be the set of vertices assigned to vertices $1, 2, \ldots, n$ respectively. The goal is to maximize the sum $\sum_{i=1}^{n} w_i(S_i)$.

	Note that sum of submodular functions is a submodular function. Hence to make the objective function of SWM fit our framework, we do the following. Let $E_u$ denote the set of edges incident to $u$. Define a function $\tilde{g}_u: E_u \rightarrow \mathbb{R}$. For any subset $S \subset E_u$, we let $\tilde{g}(S) := g(S_v)$ where $S_v$ is the set of endpoints of $S$ in $V$. The non-trivial part to handle is that in SWM any vertex in $U$ can be matched multiple times, while in our setting we allow any $U$ to be matched exactly once. We can overcome this by creating additional vertices as follows. For any $u \in U$, create $|\delta(u)|*T$ copies. A copy, indexed by $v \in \delta(u)$ and $t \in [T]$ has an edge only to $v$. Additionally, wlog we can enforce that at time $t$, an algorithm can be matched to vertices whose copy is indexed by $t$ (this doesn't change the optimal value). Therefore, we now have an instance with a submodular objective and matching constraints.
	
	Consider an arrival sequence $\tau$. Let the optimal allocation in SWM under this arrival be $S_{\tau, 1}, S_{\tau, 2}, \ldots, S_{\tau, n}$ at times $T_1, T_2, \ldots, T_{n}$. Then note that by considering the edges $(u_{S_{\tau, i}, T_{\tau, i}}, S_{\tau, i})$ yields the same value to the objective in the \osbm model. Additionally, note that the above argument also holds in the other direction since there is a one-to-one mapping between the two optimal solutions.

\subsection{Proof of Lemma~\ref{lem:lp}}
 \begin{proof}
Observe that the polytope  $\mathcal{P}$ defined by Constraints~\eqref{cons:v}-\eqref{cons:e} is downward closed. By Lemma 4 in \cite{Adam17}, we can apply continuous greedy algorithm to get an approximation solution $\x^*$ such that $\x^*$ is feasible to $\cP$ and $F(\x^*) \ge (1-1/e) \max_{\y \in \cP}f^{+}(\y)$, where $f^{+}$ is the concave closure of $f$. For each given $\y \in \cP$, by definition $f^{+}(\y)$ is the maximum value of $\E[f(\cR)]$ over all possible random sets $\cR$ satisfying the marginal distribution that $\Pr[e \in \cR] \le y_e$ for each $e$. Thus we conclude that $ \max_{\y \in \cP}f^{+}(\y)$ is an upper bound of the offline optimal and hence proving our claim.
 \end{proof}

\subsection{Proof of Theorem~\ref{thm:cra}}
\begin{proof}
Let $\Z=(Z_e)$ be the indicator vector for $e$ being matched in the \cra and let $\mathbf{X}$ be as defined in Theorem~\ref{thm:toc} with $\mathbf{x} = \mathbf{x}^*$ where $\mathbf{x}^*$ is the optimal solution to the offline program. We show that $\Z$ satisfies the two properties stated in Theorem \ref{thm:toc} with $\alp=\frac{1}{2}(1-e^{-1/2})$. Hence we have that $\E[f(\Z)] \ge \alp \E[f(\X)]$. Moreover observe that $\E[f(\X)]=F(\x^*) \ge (1-1/e) \E[\OPT]$. Combining these two facts proves Theorem~\ref{thm:cra}.

Consider an edge $e=(u,v)$ with $X_e=1$. Let $|E_{\X}(u)|=k_1$ and $|E_{\X}(v)|=k_2$, where $k_1$ and $k_2$ are the number of edges incident to $u$ and $v$ in $\X$\footnote{More precisely, we consider the sub-graph induced by the edges $e$ with $X_e=1$.} including $e$, respectively. A sufficient condition for $Z_e=1$ given $X_e=1, |E_{\X}(u)|=k_1, |E_{\X}(v)|=k_2$ is that $Y_e=1$ in this conditional space and that $Z_e=1$ given that $Y_e=1$. Thus we have,
\begin{small}
\begin{align*}
& \textstyle \Pr[Z_e=1 {~\Big |~} X_e=1, |E_{\X}(u)|=k_1, |E_{\X}(v)|=k_2] \\
&\textstyle \geq \Pr[Y_e=1 {~\Big |~} X_e=1,  |E_{\X}(u)|=k_1] * \Pr[  Z_e=1 {~\Big |~}Y_e=1, |E_{\X}(v)|=k_2, |E_{\X}(u)|=k_1] \\
& \textstyle = \frac{1}{k_1} \sum_{t=1}^T \frac{1}{T}\frac{1}{k_2}\Big(1-\frac{1}{T \cdot k_2}  \Big)^{t-1} =\frac{1}{k_1}\Big(1-\exp\big( -\frac{1}{k_2} \big) \Big). 
\end{align*}
\end{small}

In the second-last equality we used the fact that $Y_e=1$ is set uniformly for every edge $e \in E(u)$ and thus $\Pr[Y_e=1 {~\Big |~} X_e=1,  |E_{\X}(u)|=k_1] = \frac{1}{k_1}$. Moreover we claimed that $\Pr[  Z_e=1 {~\Big |~}Y_e=1, |E_{\X}(v)|=k_2, |E_{\X}(u)|=k_1] = \sum_{t=1}^T \frac{1}{T}\frac{1}{k_2}\Big(1-\frac{1}{T \cdot k_2}  \Big)^{t-1}$. The argument is as follows. Consider a given $e=(u, v)$ with $Y_e=1$. Note that from the algorithm, the only vertex $u$ can be matched to is a previous arrival of $v$. At a given time-step $t \in [T]$, this happens with probability $\tfrac{1}{k_2 |V|}$. Moreover from the integral arrival rates assumption wlog we have that $|V| = T$. Thus the probability at $u$ is safe at $t \in [T]$ is $\Big(1-\frac{1}{T \cdot k_2}  \Big)^{t-1}$. The probability that $v$ arrives at $t \in [T]$ and is assigned to $u$ is $\tfrac{1}{k_2 |V|}$. The event that $(u, v)$ is matched is mutually exclusive across the $T$ time-steps. Putting these arguments together we get that $\Pr[  Z_e=1 {~\Big |~}Y_e=1, |E_{\X}(v)|=k_2, |E_{\X}(u)|=k_1] = \sum_{t=1}^T \frac{1}{T}\frac{1}{k_2}\Big(1-\frac{1}{T \cdot k_2}  \Big)^{t-1}$.

Given the above derivation, we are now ready to prove the two properties. First we show that the \textbf{Marginal Property} holds with $\alp=(1-e^{-1/2})$. Notice that $\E[k_1]=1+\sum_{e' \in E(u), e' \neq e} x^*_{e'} \le 2$ and $\E[k_2]=1+\sum_{e' \in E(v), e' \neq e} x^*_{e'} \le r_v=2$, since $\x^*=(x^*_e)$ is feasible to \MP \eqref{LP:offline}. Also the two functions, $\frac{1}{k_1}$ and $\Big(1-\exp\big( -\frac{1}{k_2} \big) \Big)$, are convex in $k_1$ and $k_2$ respectively. Thus by Jensen's inequality, we have the following.
 \begin{small}
\begin{align*}
&\textstyle \Pr[Z_e=1| X_e=1] \textstyle \geq \E_{k_1, k_2} \Big[ \Pr[Z_e=1 {~\Big |~}  X_e=1, |E_{\X}(u)|=k_1, |E_{\X}(v)|=k_2]  \Big] \textstyle \ge \frac{1}{2}\Big(1-e^{-1/2} \Big) \doteq \alp
\end{align*}
\end{small}
The \textbf{Monotonicity Property} can be derived from the fact that $\frac{1}{k_1}$ and $\Big(1-\exp\big( -\frac{1}{k_2} \big) \Big)$ are decreasing in $k_1$ and $k_2$ respectively. Thus we are done.  
 \end{proof}

\subsection{Full Proof of Theorem \ref{thm:lpa}}
The bounds in Theorem~\ref{thm:lpa} comes from two parts. The first $(1-1/e)$ factor is obtained from approximately solving \MP \eqref{LP:offline} offline while the second $(1-1/e)$ factor is the ratio obtained because of the online nature of the problem. \emph{Note that our online analysis is tight: \lpa losses at least one factor of $(1-1/e)$ during the online stage even when $f$ is linear.} We first provide a sketch here and then prove it in detail. Recall that $\mathbf{X}$ denotes the indicator for whether an edge is matched. The notations for this proof are overloaded and hence different from that in Theorem~\ref{thm:cra}.

\xhdr{Proof sketch.} For each $u$, let $\bI_u$ indicate if $u$ is finally matched ($\bI_u=1$) after the $T$ online rounds. Set $x_u^*=\sum_{e \in E(u)} x_e^*$. We then prove the following.

\begin{lemma} \label{lem:u-side}
 $\{\bI_u: u\in U\}$ are asymptotically independent with each $\Pr[\bI_u=1]=1-e^{-x_u^*}$ when $T \rightarrow \infty$ and $|U|=o(\sqrt{T})$. 
\end{lemma}

Let $\X \in \{0,1\}^m$ be the indicator vector for $e$ being matched in $\lpa$. Let $\X_u \in \{0,1\}^m$ be the restriction of $\X$ onto the block $E(u)$ with remaining entries all zeros such that $\X=\sum_{u \in U}\X_u$. \emph{W.l.o.g.} assume that the edges are grouped by blocks of $E(u)$ for all $u \in U$: the first few entries of $\X$ corresponds to some $E(u)$ followed by another $E(u')$ and so on. Let $\sE_e$ be an $m$-dimensional unit vector with a $1$ only at coordinate $e$.  Lemma~\ref{lem:u-update} characterizes the distribution of $\X_u$. 

\vspace{0.1in}
\begin{lemma} \label{lem:u-update}
$\Pr[\X_u=\bz]=e^{-x_u^*}, ~~\Pr[\X_u=\sE_e]=(1-e^{-x_u^*})\frac{x_e^*}{x_u^*}, ~~\forall e \in E(u)$
\end{lemma}

Combining Lemmas \ref{lem:u-side} and \ref{lem:u-update}, we can view $\X$ alternatively as follows. First, start with an arbitrary vector $\mathbf{A} $ of dimension $m$. Then update the first block of $\mathbf{A}$ corresponding to some $E(u)$ by setting $\mathbf{A}_u=\bz$ with probability $e^{-x_u^*}$ and $\mathbf{A}_u=\sE_e$ with probability $\Big(1-e^{-x_u^*}\Big)\frac{x_e^*}{x_u^*}$, for every $e \in E(u)$ (denote this procedure as (*)). Now apply (*) to the remaining blocks one-by-one independently. Set $\mathbf{X}=\mathbf{A}$ at the end of this process.

In order to get the desired lower bound on $\E[f(\X)]$ we consider the following random vector $\Y$. Set $Y_e=1$ with probability $(1-1/e)x_e^*$ and with the remaining probability set it to be $0$. From Lemma 4.2 in \cite{bansal2012solving}, this implies that $\E[f(\Y)] = F(\Y) \geq (1-1/e) F(\x^*)$. The main idea of the proof is to show that after every update of (*) the invariant that the value of $\E[f(\X)]$ is at least as large as $F(\Y)$ continues to hold. Thus the final random vector $\X$ satisfies that $\E[f(\X)] \geq F(\Y)$. Thus we get $\E[f(\X)] \geq F(\Y) \geq (1-1/e) F(\x^*)$. This proves Theorem~\ref{thm:lpa} since $F(\x^*) \ge (1-1/e)\E[\OPT]$.

\xhdr{Proof of Lemma~\ref{lem:u-side}}
\begin{proof}
Consider a given $u$. Notice that in each round $t$, $u$ is assigned to a neighboring edge in $E(u)$ with probability $x^*_u/T$. Thus after $T$ online rounds it gets assigned at least once with probability $1-\Big( 1-x_u^*/T  \Big)^T \sim 1-e^{-x_u^*}$. Now we prove the statement about the independence for different $u$. Consider two given disjoint subsets $U_1$ and $U_2$ and a given vertex $u \notin U_1 \cup U_2 $, we show that 
$$\Pr\Big[\bI_u=x {~\Big |~} \bigwedge_{w \in U_1} \Big(\bI_w=1 \Big)  \bigwedge_{w \in U_2} \Big(\bI_w=0 \Big)  \Big]=\Pr[\bI_u=x], \forall x \in \{0,1\}$$

For event $\bigwedge_{w \in U_2} \Big(\bI_w=0 \Big)$ to occur, a sufficient condition is that in each round, no edge from $\bigcup_{w\in U_2} E(w)$ arrives. This occurs with probability $1-\sum_{w \in U_2}  x_u^*/T$.  For event $\bigwedge_{w \in U_1} \Big(\bI_w=1 \Big)$ to occur a sufficient condition is that in some rounds only the vertices in $U_1$ are assigned. For each $w \in U_1$, let $N_w$ be the number of assignments received by $w$ during the $T$ online rounds, \ie  the number of times $w$ gets assigned an edge by $\lpa$ (irrespective of whether it is eventually matched or not). Observe that $N_w \sim \mathrm{Pois}(x_w^*)$ when $T$ is large. We can verify that \footnote{https://eventuallyalmosteverywhere.wordpress.com/2013/02/26/poisson-tails/}  $\Pr[N_w  > \ln T] =O(T^{-2})$.

Thus with probability $1-O(|U|*T^{-2})=1-o(T^{-1})$, $N_w$ is no larger than $\ln T$ for all $w \in U_1$. Hence with a probability of  at least $1-o(T^{-1})$, the total number of rounds reserved for the shots of vertices in $U_1$ is at most $|U|*\ln T$. Thus, with probability a probability at least $1-o(T^{-1})$, we have the following.

\begin{align*}
	\left( 1-\frac{x_u^*/T}{1-\sum_{w \in U_1 \cup U_2 }x^*_w/T }  \right)^{T} & \le \Pr\Big[\bI_u=0 {~\Big |~} \bigwedge_{w \in U_1} \Big(\bI_w=1 \Big)  \bigwedge_{w \in U_2} \Big(\bI_w=0 \Big)  \Big] \\
	& \le \left( 1-\frac{x_u^*/T}{1-\sum_{w \in U_1 \cup U_2} x^*_w/T}  \right)^{T- |U|* \ln T}
\end{align*}

Note that both sides of the inequality approaches $e^{-x_u^*}$ when $|U|=o(\sqrt{T})$ and $T \rightarrow \infty$. Thus we are done. 
\end{proof}

 \noindent\textbf{Proof of Lemma~\ref{lem:u-update}}
\begin{proof}
The first part of the Lemma follows directly from Lemma \ref{lem:u-side} with $\bI_u=0$. We will now prove the second part. Let $\bI_{u,t}$ indicate that $u$ is matched (for the first time) at time $t$. Consider a given $e \in E(u)$. 
\begin{align*}
\Pr[ \X_u=\sE_e | \bI_u=1] &= \sum_{t=1}^T \Pr[ \X_u=\sE_e | \bI_{u,t}=1]\cdot \Pr[  \bI_{u,t}=1| \bI_u=1]\\
&= \sum_{t=1}^T \frac{x_e^*/T}{\sum_{e'\in E(u)}x_{e'}^*/T} \cdot \Pr[  \bI_{u,t}=1| \bI_u=1]\\
&= \frac{x_e^*}{x_u^*}
\end{align*}
The second equality follows from the following argument. In each time $t$, $u$ will be assigned an edge $e' \in E(u)$ with probability $x_{e'}^*/T$. Thus conditioning on the event that $u$ gets matched at $t$, the conditional probability that $u$ gets matched by edge $e$ is exactly $\frac{x_e^*/T}{\sum_{e'\in E(u)}x_{e'}^*/T}$. Therefore we have 
	\[
	\Pr[\X_u=\sE_e]=\Pr[\X_u=\sE_e | \bI_u=1] \cdot \Pr[\bI_u=1]=\Big(1-e^{-x_u^*}\Big)\frac{x_e^*}{x_u^*}.
	\]
\end{proof}

%that the final indicator random vector $\X$ can be viewed in the following constructive way parameterized by $\x^*$: 

%Now we are ready to prove the key theorem.
%Let $\alp=(1-1/e)$. We will try to prove that $\E[f(\X)] \ge F(\alp )$
%%the set of all possible realizations if we sample each entry independently according to $\Y_{-u}$.  
\begin{theorem} \label{thm:key}
$\E[f(\X)] \ge (1-1/e) F(\x^*)$
\end{theorem}

\begin{proof}
Let $\Y, \mathbf{A}$ be as defined in the proof sketch above. Let $\Z$ be the random vector obtained after updating the first block of $\mathbf{A}$, say $E(u)$. We show that $\E[f(\Z)] \ge F(\mathbf{A})$. Let $\mathbf{R} = (1-1/e)\mathbf{x}^*$.

Let $\mathbf{R}_{-u} \in [0,1]^m$ be the restriction of $\mathbf{R}$ onto all blocks other than $E(u)$ (hence the entries in the block $E(u)$ are all zeroes). Note that both $\mathbf{R}$ and $\mathbf{R}_{-u}$ have the same length $m$. Let $\cR_{-u}$ be the random indicator vector to denote whether an edge was sampled (independently with probability according to $\mathbf{R}_{-u}$).

% and let $\cS(\Y_{-u}) \subseteq 2^{\{0,1\}^m}$ be set of all possible realizations of $\cY_{-u}$.
%\in \cS(\Y_{-u})  \in \cS(\Y_{-u}) 

From the above analysis, we have the following.
\begin{align}
&\E[f(\Z)] \\
&= \textstyle \sum_{e \in E(u)} \Big(1-e^{-x_u^*}\Big)\frac{x_e^*}{x_u^*} F(\mathbf{R}_{-u}+\sE_e)+e^{-x_u^*}F(\mathbf{R}_{-u}) \label{eqn:lpa1}\\
& \ge \textstyle \sum_{e\in E(u)} (1-1/e) x_e^* F(\mathbf{R}_{-u}+\sE_e)+\Big(1-(1-1/e) x_u^*\Big)F(\mathbf{R}_{-u})  \label{eqn:lpa2}\\
&= \textstyle \sum_{e \in E(u)}  (1-1/e) x_e^* \Big( \sum_{\ba }\Pr[\cR_{-u}=\ba] F(\ba+\sE_e) \Big) \nonumber\\ &+\Big(1-(1-1/e) x_u^*\Big)  \Big( \sum_{\ba  }\Pr[\cR_{-u}=\ba] F(\ba) \Big)  \label{eqn:lpa3}\\
&= \textstyle \sum_{\ba}\Pr[\cR_{-u}=\ba] \Big( 
\sum_{e \in E(u)}  (1-1/e) x_e^*  F(\ba+\sE_e)+\big( 1-(1-1/e) x_u^* \big) F(\ba) \Big)  \label{eqn:lpa4}
\end{align}

Inequality \eqref{eqn:lpa2} is due to the monotonicity of $f$:  we have that $F$ is non-decreasing over each entry and thus $F(\mathbf{R}_{-u}+\sE_{e}) \ge F(\mathbf{R}_{-u})$ for all $e$; note that the probability mass distributed over each $F(\mathbf{R}_{-u}+\sE_e)$ is $\Big(1-e^{-x_u^*}\Big)\frac{x_e^*}{x_u^*}  \ge (1-1/e) x_e^*$ and thus by squeezing the surplus mass from each lower bound to the last item $F(\mathbf{R}_{-u})$, we get our claim. Inequality \eqref{eqn:lpa3} follows directly from the definition of the multilinear extension $F$.

% binary vector $\ba \in  \cS(\Y_{-u})$
For a given realization $\ba$ of $\cR_{-u}$, define a set function $g_{\ba}$ over $E(u)$ as follows. Consider a given binary vector $\bb' \in \{0,1\}^{|E(u)|}$. For notational convenience, we extend $\bb'$ to a vector of length of $m$ by filling the remaining entries as $0$; let $\bb$ be the resultant vector. Define $g_{\ba}(\bb')=f(\ba+\bb)$ for each $\bb' \in \{0,1\}^{|E(u)|}$. From  Lemma~\ref{lem:lpa_sub_g}, we have that each given $\ba$, $g_{\ba}$ is a monotone submodular set function over $E(u)$. Let $\mathbf{R}_u \in [0,1]^m$ be the restriction of $\mathbf{R}$ onto the block $E(u)$ with the remaining entires being $0$ and $\mathbf{R}'_u \in [0,1]^{|E(u)|}$ be the truncated version of $\mathbf{R}_u$ onto $E(u)$. Let $G_{\ba}$ be the multilinear extension of $g_{\ba}$. Focus on the second part of the expression in \eqref{eqn:lpa4}.

\begin{small}
\begin{align}
\Lambda & \doteq \sum_{e \in E(u)}  (1-1/e) x_e^*  F(\ba+\sE_e)+\big( 1-(1-1/e) x_u^* \big) F(\ba)  \label{eqn:lpa_a1}\\
 &=\sum_{e \in E(u)}  (1-1/e) x_e^*  g_{\ba}(\sE'_e)+ \big( 1-(1-1/e) x_u^* \big) g_{\ba}(\bz)  \label{eqn:lpa_a2}\\
  &\ge G_{\ba}(\Y'_u)  \label{eqn:lpa_a3}  \\
  &=\sum_{\bb' }\Pr[\cR'_{u}=\bb'] g_{\ba}(\bb')
  \label{eqn:lpa_a4}\\
  &=\sum_{\bb }\Pr[\cR_{u}=\bb] f(\ba+\bb)  \label{eqn:lpa_a5}
\end{align}
\end{small}

In \eqref{eqn:lpa_a2}, both $\sE'_e$ (elementary vector of $e$) and $\bz$ have the same length $|E(u)|$. Recall that $\mathbf{R}=(1-1/e)\x^*$ and thus $\mathbf{R}'_u=\Big((1-/e)x_e^*: u\in E(u)\Big)$. Inequality  \eqref{eqn:lpa_a3} follows from applying Lemma \ref{lem:lpa_g_neg},as shown later, to the function $g_{\ba}$ and $\mathbf{R}'_u$. In equalities \eqref{eqn:lpa_a4} and \eqref{eqn:lpa_a5},  $\cR'_u \in \{0,1\}^{|E(u)|}$  and $\cR_u \in \{0,1\}^m$ are the  respective random indicator vectors denoting if an edge was sampled (independently according to $\mathbf{R}'_u$ and $\mathbf{R}_u$). The two equalities  \eqref{eqn:lpa_a4} and \eqref{eqn:lpa_a5} follow directly from the definitions of $g_{\ba}$ and $G_{\ba}$.   

Substituting the result in \eqref{eqn:lpa_a5} back to \eqref{eqn:lpa4}, we have that 
\begin{align*}
\E[f(\Z)]  \ge  \sum_{\ba}\Pr[\cR_{-u}=\ba] \sum_{\bb }\Pr[\cR_{u}=\bb] f(\ba+\bb) =\E[f(\Y)] = F(\Y).
\end{align*}

By applying the above analysis repeatedly to each block $E(u)$, we have that the final random integral vector $\Z$ is same as $\X$ and that it satisfies $\E[f(\X)] \ge F(\Y) \ge (1-1/e) F(\x^*)$.   
\end{proof}

%\in \cS(\Y'_{u}) 
%\in \cS(\Y_{u}) 
%and $\cS(\Y'_u)$ ($\cS(\Y_u)$) is the resuntant 

% refers to that elementary vector of $e$ with length $|E(u)|$ and Inequality \eqref{eqn:lpa_a3} is due to Lemma XX by applying monotone submodular function $g_{\ba}$ over the vector $\Y'_u=\Big((1-/e)x_e^*: u\in E(u)\Big)$. Let $\cY_u$ be the random indicator vector if we sample each edge in $E(u)$ independently according to $\Y_u$ and $\cS(\Y_u)$ be the set of all possible realizations of $\cY_u$. 

%By Lemma XX, we have that 
%\begin{align}
% \sum_{e \in E(u)}  (1-1/e) x_e^*  F(\ba+\sE_e)+\big( 1-(1-1/e) x_u^* \big) F(\ba) 
% &=\sum_{e \in E(u)}  (1-1/e) x_e^*  g_{\ba}(\sE_e)+ \big( 1-(1-1/e) x_u^* \big) g_{\ba}(\bz) \\
% & \ge G_{\ba}(\x)
%\end{align}

 \subsection{Technical Lemmas}
 
Throughout this section, we assume $f$ is a general non-negative monotone submodular function over the ground set $[m]=\{1,2,\cdots, m\}$. Consider a given subset $S \subseteq [m] $ and define $g_{S}$ as a set function over $[m]-S$ as follows: for each $S' \subseteq [m]-S$, $g_{S}(S')=f(S \cup S')$. 
%\vspace{0.1in}
\begin{lemma} \label{lem:lpa_sub_g}
For each given $S\subseteq [m]$, $g_S$ is a monotone submodular function over $[m]-S$.
\end{lemma}

\begin{proof}
Consider any two subsets $S'_1, S'_2 \subseteq [m]-S$. Observe that 
\begin{align*}
g_S(S'_1)+g_S(S'_2) &= f(S \cup S'_1)+ f(S \cup S'_2)\\
& \ge f\Big(S \bigcup \Big( S'_1  \cup S'_2 \Big) \Big)+f\Big(S \bigcup \Big( S'_1  \cap S'_2 \Big) \Big) \\
&=g_S( S'_1  \cup S'_2)+g_S( S'_1  \cap S'_2)
\end{align*}

Thus by definition, we have that $g_S$ is submodular. The monotonicity of $g_S$ follows directly from that of $f$. 
\end{proof}

Let $F$ be the the multilinear extension of $f$. Consider a given vector $\x \in [0,1]^m$ such that $\sum_{i \in [m]} x_i \le 1$. Let $\sE_i \in \{0,1\}^m$ be the standard unit  vector such that only the entry at position $i$ is $1$ and all rest are $0$.
\vspace{0.1in}
\begin{lemma}\label{lem:lpa_g_neg}
$F(\x) \le \sum_{i=1}^m x_i f(\sE_i)+(1-\sum_{i =1}^m x_i) f(\bz) $
\end{lemma}

\begin{proof}
We prove by induction on $m$. For the base case $m=1$, we can verify that the inequality becomes tight from the definition of $F$. Now assume our claim is valid for all $m \le k-1$. Consider the case $m=k$.

Let $\x_{-k} \in [0,1]^k$ be the restriction of $\x$ onto the first $k-1$ coordinates with the last entry being $0$. Similarly let $\x_k \in [0,1]^k$ be the restriction of $\x$ onto the last coordinate. Note that both $\x_{-k}$ and $\x_k$ have the same length $\x$ such that $\x=\x_{-k}+\x_k$. We then have the following, which completes the proof.
\begin{align}
F(\x) &= F(\x_{-k}+\x_k)  \label{eqn:app_lpa_1}\\
&= x_k F(\x_{-k}+\sE_k) +(1-x_k) F(\x_{-k})  \label{eqn:app_lpa_2}\\
&\le x_k \Big(  F(\x_{-k})+f(\sE_k)-f(\bz) \Big)+(1-x_k) F(\x_{-k})  \label{eqn:app_lpa_3}\\
&=F(\x_{-k})+x_k \Big( f(\sE_k)-f(\bz) \Big)  \label{eqn:app_lpa_4}\\
& \le \sum_{i=1}^{k-1} x_i f(\sE_i)+ (1-\sum_{i=1}^{k-1} x_i) f(\bz)+x_k \Big( f(\sE_k)-f(\bz) \Big)  \label{eqn:app_lpa_5}\\
&= \sum_{i=1}^{k} x_i f(\sE_i)+ (1-\sum_{i=1}^{k} x_i) f(\bz) \label{eqn:app_lpa_6}
\end{align}

Equality \eqref{eqn:app_lpa_2} follows from the definition of $F$; Inequality\eqref{eqn:app_lpa_3} can be proved as follows. Let $\cX_{-k}$ be the random indicate vector to denote if an element was sampled (independently according to $\x_{-k}$). For every realization $\cX_{-k}=\ba \in \{0,1\}^k$, we have that $f(\ba+\sE_k) \le f(\ba)+f(\sE_k)-f(\bz)$, since $\ba$ and $\sE_k$ represents two disjoint sets over $[k]$. Thus by taking expectation over $\cX_{-k}$, we get the Inequality~\eqref{eqn:app_lpa_3}. Inequality \eqref{eqn:app_lpa_5} follows from the inductive hypothesis.
\end{proof}

\section{Supplementary section for experiments}

	\subsection{Mathematical Program for the experiments in main section}
		Recall that with every vertex $v$, there was a weighted coverage function $f_v: 2^E \rightarrow \mathbb{R}$ associated. The objective was $\sum_{v \in V} f_v$. Our offline program can be reformulated using the epigraph form of the mathematical program and reduces to the following Linear Program. Let $[g]$ denote the set of genres and we use $z$ to index into each genre.
			
			\begin{alignat}{2}
				\label{LP:MovieLens}
				\text{maximize}    & ~~\textstyle \sum_{v \in V} \sum_{z \in [g]} w_{z, v} \gamma_{z, v} & \\
				\text{subject to}  & \textstyle \sum_{e \in \nr{v} } x_{e} \le  \eta*r_v    & \textstyle  \qquad \forall v \in V\\
                	&\textstyle  \sum_{e \in \nr{u} } x_{e} \le  B    & \textstyle \qquad \forall u \in U\\
                   & \textstyle 0 \le x_{e} \le 1   &\textstyle \forall e \in E\\
					&\textstyle \gamma_{z, v} \leq \smashoperator{\sum_{e \in \delta(v): q_e[z]=1}} x_e & \textstyle \forall z \in [g], v \in V\\ 
                   & \textstyle \gamma_{z, v} \leq 1 & \textstyle \forall z \in [g], v \in V
				\end{alignat}		

	Since \eqref{LP:MovieLens} is a linear program, we can use standard tools to solve this linear program exactly.
	
	\subsection{Additional experiment for MovieLens dataset with integral arrival rates}
	
		Here we describe an additional experiment on the MovieLens dataset. We consider the case of $\eta=1$ and compare the performance of various algorithms. The experimental setting is the same as that described in the main section. The key difference is that here we let every user come with probability $1/T$ at each time-step. Hence $T=|V|=200$ in these experiments. Figure~\ref{fig:intRates} plots the performance for varying $b$. As seen in the plot, both \lpa and \cra comprehensively beats the baselines under this setting.

		\begin{figure*}[h!]
    		\centering
        \includegraphics[scale=0.27]{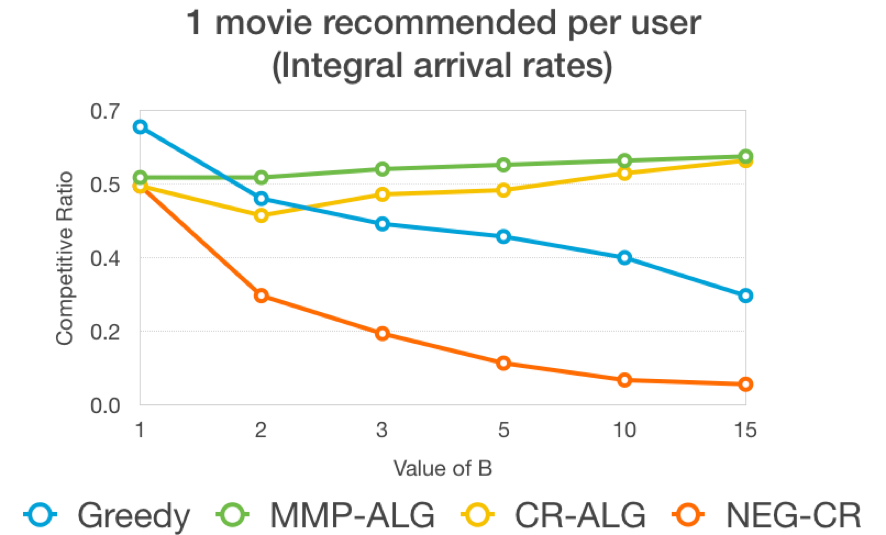}
        \caption{Comparing algorithms with $\eta=1$ on MovieLens dataset with integral arrival rates (\ie $r_v=1$)}
        	\label{fig:intRates}
        \vspace{1cm}	
		\end{figure*}

	\subsection{Experiments on Synthetic Data}
	We now describe our experimental results on synthetic data. The landscape of this problem is vast and we do not aim to be exhaustive. Instead, we focus on two important submodular functions, namely the budget-additive function and the coverage function (both defined in the preliminaries). We consider the same setting as the experiments in the main section. For both these special cases, the offline program can be reduced to a linear program (\LP) which can be solved efficiently.

	\xhdr{Offline programs.}
	We will now describe the offline programs for these two objectives. In our experiments, we use the optimal solution of these respective programs for both benchmark of the offline as well as a guide to our online algorithms.
	
	\begin{enumerate}[noitemsep]
		\item \emph{Budget-additive function.} We want to solve the following mathematical program as our offline benchmark.
				\begin{alignat}{2}
				\label{LP:budgetAdditive_act}
				\text{maximize}    & ~~\textstyle \min\{B, \sum_{e \in E} w_e x_e\}  & \\
				\text{subject to}  &  \textstyle \sum_{e \in \nr{v} } x_{e} \le  r_v    &  \qquad \forall v \in V\\
                	& \textstyle \sum_{e \in \nr{u} } x_{e} \le  1    & \textstyle \qquad \forall u \in U\\
                   & \textstyle 0 \le x_{e} \le 1   & \textstyle \forall e \in E
				\end{alignat}
			This can be converted to a \LP using a standard transformation (similar to the epigraph form of convex programs) as follows.

			\begin{alignat}{2}
				\label{LP:budgetAdditive}
				\text{maximize}    & ~~ \gamma & \\
				\text{subject to}  &\textstyle  \sum_{e \in \nr{v} } x_{e} \le  r_v    & \textstyle \qquad \forall v \in V\\
                	& \textstyle \sum_{e \in \nr{u} } x_{e} \le  1    & \textstyle \qquad \forall u \in U\\
                   & \textstyle 0 \le x_{e} \le 1   & \forall e \in E\\
					&\textstyle  \gamma \leq \sum_{e \in E} w_e x_e &\\ 
                   & \textstyle \gamma \leq B & 
				\end{alignat}
		\item \emph{Coverage function.} Similar to above we can convert it to a standard \LP form and here we directly describe the final form. Every edge is associated with a binary feature-vector $\mathbf{q}_e$ of dimension $g$. Weight of the $z^{th}$ dimension is represented as $w_z$.

				\begin{alignat}{2}
				\label{LP:coverage}
				\text{maximize}    & ~~\textstyle \sum_{z \in [g]} w_z \gamma_z & \\
				\text{subject to}  & \textstyle \sum_{e \in \nr{v} } x_{e} \le  r_v    & \textstyle  \qquad \forall v \in V\\
                	&\textstyle  \sum_{e \in \nr{u} } x_{e} \le  1    & \textstyle \qquad \forall u \in U\\
                   & \textstyle 0 \le x_{e} \le 1   &\textstyle \forall e \in E\\
					&\textstyle \gamma_z \leq \smashoperator{\sum_{e \in E: q_e[z]=1}} x_e & \textstyle \forall z \in [g]\\ 
                   & \textstyle \gamma_z \leq 1 & \textstyle \forall z \in [g]
				\end{alignat}
		
	\end{enumerate}
	
	\begin{figure*}[h!]
    \centering
    \begin{subfigure}[t]{0.45\textwidth}
        \centering
        \includegraphics[scale=0.27]{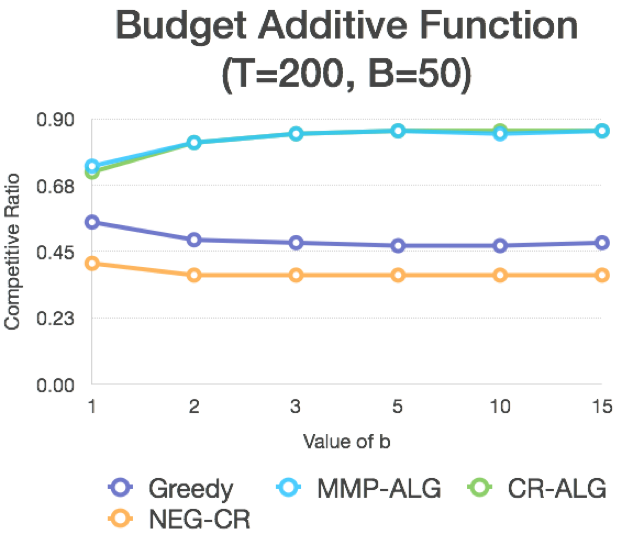}
        \caption{Comparing algorithms on the budget-additive function. $x$ and $y$-axes represent the $b$ value in $b$-matching and competitive ratio respectively.}
        \label{fig:exp1}
    \end{subfigure}~
    \begin{subfigure}[t]{0.45\textwidth}
        \centering
        \includegraphics[scale=0.27]{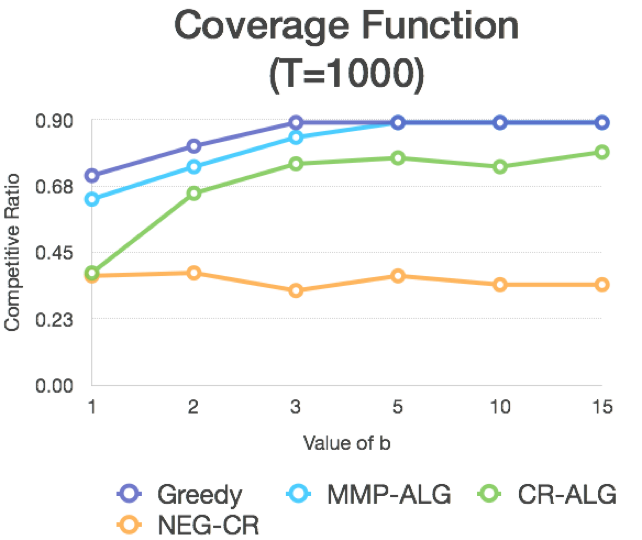}
                \caption{Comparing algorithms on the coverage function.}
        \label{fig:exp2}
    \end{subfigure}
    \vspace{1cm}	
\end{figure*}

	\xhdr{Dataset.} We use a simulated dataset for our experimental purposes as follows. First simulate an instance of the graph and then simulate the arrival sequence. For both objectives, we associate the online type with an arrival rate chosen uniformly at random in the range $[0, 1]$. Then simulate the arrival process based on these rates. The edges are chosen as follows. For every online type, we set a random set of at most $10$ vertices in the offline side as its neighbor and edges between this online type and those offline vertices. We consider the $b$-matching scenario and vary $b$ in the set $\{1, 2, 3, 5, 10, 15\}$. The cardinalities of offline vertices and online vertices vary based on the experiment which we will describe below.
	
	\begin{enumerate}[noitemsep, topsep=0pt]
		\item \emph{Coverage function.} For this experiment, the number of offline vertices to be $40$ and the number of online types to be $200$. The time-horizon is set to $1000$ (hence a lot of time-steps has no arrival of a online vertex). Each offline and online type is associated with a random subset, of size at most 10, of $\{1, 2, \ldots, 1000\}$ as the set of features. For every edge $e=(u, v)$, the feature vector $\mathbb{q}_e$ associated with it is given by the union of the features of its end-points $u$ and $v$. Every feature in $[1000]$ is associated with a weight chosen uniformly at random in the interval $[0, 1]$.
		\item \emph{Budget-additive function.} In this experiment, the number of offline vertices to be $100$ and the number of online types to be $200$. The time-horizon is $200$. For every edge $e$ we choose a random number uniformly in $[0, 1]$ and set the weight $w_e$ as this random number. We set the budget $B=50$.
		
	\end{enumerate}
	
	%%%% (3) Description of the various algorithms %%%%%%

	\COMM{
	\xhdr{Algorithms.} We test both our CR-based (Algorithm~\ref{alg:cr}) and the \MP-based (Algorithm~\ref{alg:nadap}) experimentally. Additionally we consider the following two heuristics for our study\footnote{Both these can also be implemented in the value-oracle model.}. (1) \emph{Greedy:} We consider the following natural greedy algorithm. At time step $t$, let $S_t$ be the set of edges chosen so far. When a vertex $v$ arrives, we choose an available neighbor $u$ that maximizes $f(S_t \cup \{(u, v)\})-f(S_t)$. If no neighbor is available, we drop $v$. (2) \emph{Negative CR-based algorithm (\NEGCR):} This algorithm is a small tweak to the CR-based algorithm. We replace the independent sampling followed by the uniform sampling at every $u$ steps with the following procedure. At every $u$, we use the dependent rounding routine due to \cite{gandhi2006dependent} to obtain a semi-matching $\mathcal{M}_1$\footnote{Each $u \in U$ has degree $1$ but some $v \in V$ can have degree greater than $1$.}. In the online phase when a $v$ arrives, we sample one of its available neighbors in $\mathcal{M}_1$ uniformly and match it. If not, we drop $v$.
	}
	
	%%%% (4) Comparison of the various algorithms and discussion %%%%%%

	\xhdr{Results.} 
		Figures~\ref{fig:exp1} and \ref{fig:exp2} plot the results of the experiments on the budget-additive function and coverage function, respectively. In budget-additive case, our algorithms $\lpa$ and $\cra$ significantly outperform the heuristic baselines, namely, greedy and \NEGCR. Additionally the competitive ratios of our algorithm are much higher than the theoretical benchmark ($0.63$ and $0.20$ respectively). Note that as $b$ increases for our experimental parameters the competitive ratios only slightly increase. In the coverage function scenario, we see that greedy out-performs both our algorithms when $b$ is small. But as soon as $b$ is slightly increased, our algorithm matches the performance of Greedy. This suggests that for very special coverage functions one might be able to construct algorithms that do better. However, even on such cases our general algorithm is \emph{very competitive}. In this scenario, increasing the value of $b$ has a larger effect than before, on the competitive ratio, yet not drastically larger. These experiments are a further evidence to support that analysis on $b=1$ gives a good proxy for understanding the performance on larger $b$.

\end{document}